\documentclass[twoside,11pt]{article}

\makeatletter
\let\Ginclude@graphics\@org@Ginclude@graphics
\makeatother

\usepackage[preprint]{jmlr2e}

\usepackage{microtype}
\usepackage{graphicx}
\usepackage{subfigure}
\usepackage{booktabs} 

\usepackage{epsfig,epsf,fancybox}
\usepackage{amsmath}
\usepackage{mathrsfs}
\usepackage{amssymb}
\usepackage{amsfonts}
\usepackage{graphicx}
\usepackage{color}
\usepackage{stmaryrd}
\usepackage{multirow}
\usepackage{booktabs}
\usepackage{float}
\usepackage{bbm}

\usepackage{algorithm}
\usepackage{algorithmic}

\def\proof{\par\noindent{\em Proof.}}

\usepackage{mathtools}

\usepackage[normalem]{ulem} 

\newcommand{\bc}{\begin{center}}
	\newcommand{\ec}{\end{center}}

\newcommand{\bdm}{\begin{displaymath}}
	\newcommand{\edm}{\end{displaymath}}

\newcommand{\beq}{\begin{equation}}
	\newcommand{\eeq}{\end{equation}}

\newcommand{\bfl}{\begin{flushleft}}
	\newcommand{\efl}{\end{flushleft}}

\newcommand{\bt}{\begin{tabbing}}
	\newcommand{\et}{\end{tabbing}}

\newcommand{\beqn}{\begin{eqnarray}}
	\newcommand{\eeqn}{\end{eqnarray}}

\newcommand{\beqs}{\begin{align*}} 
	\newcommand{\eeqs}{\end{align*}}  

\jmlrheading{}{2021}{}{}{}{}{Gao, Tang, Xu}

\ShortHeadings{}{Gao, Tang, Xu}
\firstpageno{1}

\begin{document}
	
	\title{Optimal Compression for Minimizing Classification Error Probability: an Information-Theoretic Approach}
	
	\author{\name Jingchao Gao \email jingchao-gao@uiowa.edu\\
		\addr Department of Mathematics\\
		The University of Iowa\\
		Iowa City, IA 52242, USA
		\AND
		\name Ao Tang \email atang@cornell.edu \\
		\addr Electrical and Computer Engineering\\
		Cornell University\\
		Ithaca, NY 14850, USA
        \AND 
        \name Weiyu Xu \email weiyu-xu@uiowa.edu\\
        \addr Electrical and Computer Engineering\\ 
        The University of Iowa\\ 
        Iowa City, IA 52242, USA}
	
\editor{}
	
	\maketitle
\begin{abstract}
We formulate the problem of performing optimal data compression under the constraints that compressed data can be used for accurate classification in machine learning. We show that this translates to a problem of minimizing the mutual information between data and its compressed version under the constraint on error probability of classification is small when using the compressed data for machine learning. We then provide analytical and computational methods to characterize the optimal trade-off between data compression and classification error probability. First, we provide an analytical characterization for the optimal compression strategy for data with binary labels. Second, for data with multiple labels, we formulate a set of convex optimization problems to characterize the optimal tradeoff, from which the optimal trade-off between the classification error and compression efficiency can be obtained by numerically solving the formulated optimization problems. We further show the improvements of our formulations over the information-bottleneck methods in classification performance.   

\end{abstract}
\section{Introduction}
\label{sec:intro}
 Machine learning plays an important role in science and engineering.  Among machine learning tasks, classification is an important one which has many applications in communication and signal processing, for example, image recognition. 

 Machine learning needs sensor data to make inference or to perform classification \citep{murphy2013machine,bishop:2006:PRML}. These sensor data are first collected, and then stored in storage or transmitted through communication channels to classifiers. However, the capacities of storage or communication channel are often limited. Thus, there is often a need to compress sensing data for more efficient storage or transmission \citep{Calderbank09compressedlearning:,ZISSELMAN20183,cheng2021datasharing}. A fundamental question is hence how much compression one can achieve for sensing data such that machine learning tasks can still be executed with a certain given accuracy? In this paper, we propose a formulation of this problem, and try to answer this question for classification from an information-theoretic perspective.  
 
  In classification, we assume that labels (denoted by random variable $Y$) generate data (denoted by $X$) according to data generation distribution $\text{P}(X|Y)$. Data $X$ is fully  known to the data compressor.  The data compressor compresses $X$ into compressed data $\tilde{X}$.  The goal for the compressor is to compress $X$ as much as possible for efficient communication or storage while allowing the classification task to be performed still with a specified fidelity: namely the label $Y$ can still be sufficiently accurately recovered using only compressed data $\tilde{X}$. Towards this end, we propose to minimize the mutual information between  $X$ and $\tilde{X}$ while minimizing the error probability (or generalized costs associated with classification errors).
 
 In classical rate-distortion theory for lossy data compression, data compression is performed so that the mutual information between data $X$ and compressed data $\tilde{X}$ is minimized  under the constraint on a distortion criterion between $X$ and $\tilde{X}$ \citep{cover}. The distortion criterion in rate-distortion theory is often a direct distortion measure depending on the original data $X$ and the compressed data $\tilde{X}$. In contrast, in this paper, for the classification task, we are considering the distortion between the original label and the recovered label ($\hat{Y}$) for classification, rather than the direct distortion between $X$ and $\tilde{X}$.
 
  Our research problem is connected with the information bottleneck principle \citep{Tisbybottle}\citep{tishby2015learning}\citep{bottle}\citep{entropy}, which was proposed to study data compression under the constraint of preserving classification labels to a certain fidelity. The information bottleneck principle uses the mutual information between label ($Y$) and compressed data ($\tilde{X}$) as a simple proxy for the fidelity in preserving the label information.  However, mutual information may not be an accurate indicator of the distortion between the recovered label $\hat{Y}$ and the original label $Y$ in the classification task. This is especially true if the distortion in classification is asymmetric: the distortion for mis-classifying an object with label ``$a$'' to label ``$b$'' is weighted  higher than mis-classifying an object with label ``$b$'' to label ``$a$''.  In addition, \citet{raginskyDBLP:journals/corr/abs-0901-1905,raginsky6353589} looked at rate-limited communication of training data in machine learning and derived performance limits of constructed predictors based on such rate-limited communication.

 In this paper, we directly consider more relevant metrics for characterizing classification performances in determining optimal compression of sensing data. In particular, we study the problem of minimizing the mutual information between data and compressed data under constraints on classification error probability (or or generalized costs associated with classification errors), which are widely used performance metric for evaluating a classifier. 
  The rest of this paper is organized as follows. In Section \ref{sec:format}, we formulate the problem of optimally compressing data under classification error probability constraints. In Section \ref{sec:binary}, we analytically characterize the optimal compression strategy for binary symmetric channel connecting label and sensing data. In Section \ref{sec:general}, we propose a general optimization framework to calculate the optimal compression and resulting minimum classification error probability. In Section \ref{sec:numerical}, we present numerical results showing the optimal trade-off between data compression and classification error probability. 
 
\section{Model Formulation}
\label{sec:format}
Suppose that we have $m$ labels in the label set $\mathcal{Y}$, which is $\{y_{1}, y_{2},\ldots, y_{m}\}$. We let the prior probability for the labels be $\text{P}(y_{i})$, $i=1, 2, \ldots, m$. 
Then the label ($Y$) will generate data, and we denote the set of possible data as $\mathcal{X}$. We assume that  $\mathcal{X}$ has $n$ elements, and its elements are $x_{1}, x_{2}, \ldots, x_{n}$. We denote the transition probability between each label and any possible data as $\text{P}(x_{j}|y_{i})$, where $i=1, 2, \ldots, m;$ and $j=1, 2, \ldots, n$. For efficient storage and communication, we want to compress data $X$ to compressed data $\tilde{X}$, which are sampled from set $\tilde{\mathcal{X}}$ of cardinality $l$. To be exact, $\tilde{\mathcal{X}}$ includes $\tilde{x}_{1}, \tilde{x}_{2}, \ldots, \tilde{x}_{l}$ as its elements.  Furthermore, we define that the transition probability between each data $X$ and its compressed data $\tilde{X}$ as $\text{P}(\tilde{x}_{k}|x_{j})$, where $j=1, 2, \ldots, n,$ and $k=1, 2, \ldots, l.$
\begin{figure}
\begin{center}
  \includegraphics[width=10cm]{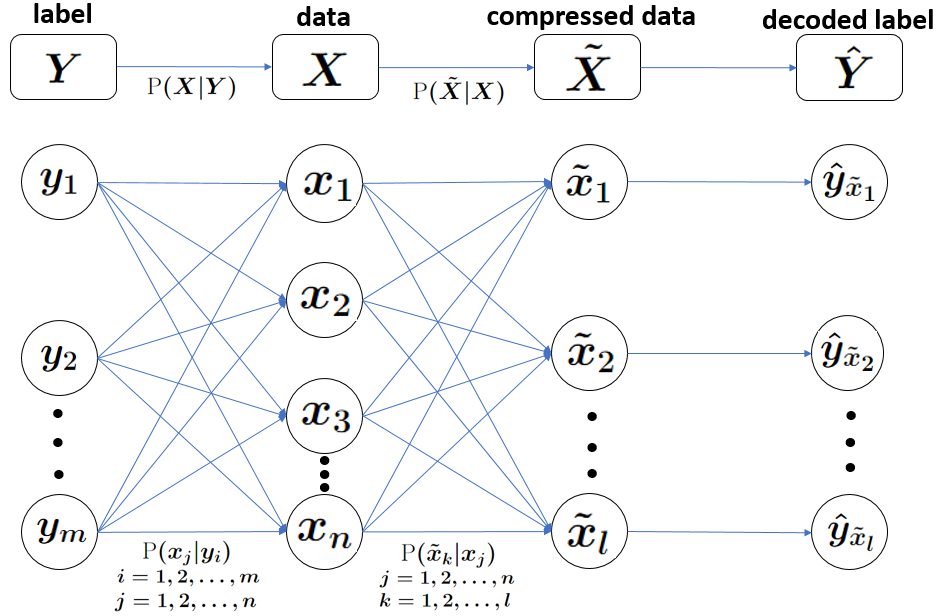}
  \caption{Transition probabilities between labels, data, and compressed data}
  \end{center}
\end{figure}
We assume that the decoder or machine learning algorithms use the maximum a posteriori (MAP) decoder (or the minimum-cost decoder when general costs associated with decoding errors are considered) to decode compressed data $\tilde{x_k}$ to label $\hat{y}_{\tilde{x}_k}$, where $1\leq k \leq l$. The job of the compressor is to design the transition probabilities $\text{P}(\tilde{x}_{k}|x_{j})$'s such that the mutual information $I(X, \tilde{X})$ is minimized for most efficient compression, while keeping the decoding error probability (the probability that the decoded label is not equal to the original label) smaller than a certain threshold. 

\section{Optimal Compression for Binary Symmetric Channel: Analytical Results}
\label{sec:binary}
While it is difficult to obtain analytical solutions to the proposed problem in general, we are able to analytically derive analytical optimal compression strategies for binary labels and data. We consider the case of binary labels and we assume that there are also two elements in the alphabet for data and the alphabet for compressed data. We assume that $\text{P}(Y=0)=\frac{1}{2}$, and $\text{P}(Y=1)=\frac{1}{2}$. We try to minimize the mutual information between $X$ and $\tilde{X}$ (subject to MAP decoding error threshold constraints) over the following transition probabilities $p_1$, $p_2$ and $p_3$: $\text{P}(X=1|Y=0)=\text{P}(X=0|Y=1)=p_{1}$, $\text{P}(\tilde{X}=1|X=0)=p_{2}$, and $\text{P}(\tilde{X}=0|X=1)=p_{3}$.\\
\begin{figure}
\begin{center}
  \includegraphics[width=10cm]{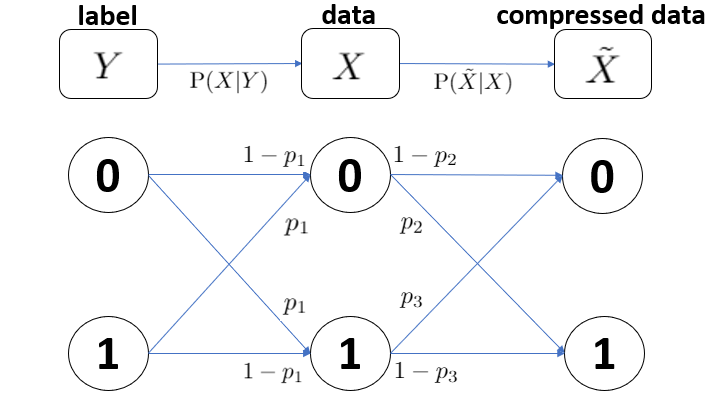}
\caption{Transition probabilities for binary data.}
\end{center}
\end{figure}
\begin{theorem}
For binary data, where each label has equal probability, and with symmetric crossover transition probabilities that are less than $\frac{1}{2}$ between label and data, the optimal trade-off in terms of classification error probability and data compression is achieved by having symmetric transition probabilities between data and compressed data (namely $p_2=p_3 \leq \frac{1}{2}$). Then the smallest achievable mutual information between $X$ and $\tilde{X}$ is $I=1-p_{2}\log\frac{1}{p_{2}}-(1-p_{2})\log\frac{1}{1-p_{2}}$ corresponding to an error probability no bigger than $Pe=p_{1}+p_{2}-2p_{1}p_{2}$.  
\end{theorem}
\begin{proof}
In this proof, we show that if $p_2 \neq p_3$, we can always make the crossover probability symmetric and equal to the average of $p_2$ and $p_3$, without increasing $I(X, \tilde{X})$ and without increasing the MAP decoding error probability.

\begin{equation*}
        \begin{split}
            \text{P}(Y=1|\tilde{X}=0)&=\frac{\text{P}(Y=1,\tilde{X}=0)}{\text{P}(\tilde{X}=0)}\\
            &=\frac{p_{1}+p_{3}-p_{1}p_{2}-p_{1}p_{3}}{1-p_{2}+p_{3}},
        \end{split}
    \end{equation*}
    Similarly,
    \begin{equation*}
    \text{P}(Y=0|\tilde{X}=1)=\frac{p_{1}+p_{2}-p_{1}p_{3}-p_{1}p_{2}}{1+p_{2}-p_{3}}.
    \end{equation*}
    Then, $\text{P}(Y=0|\tilde{X}=0)=1-\text{P}(Y=1|\tilde{X}=0)$, $\text{P}(Y=1|\tilde{X}=1)=\text{P}(Y=0|\tilde{X}=1)$.
    Since $p_{1}<\frac{1}{2}$ and $p_{1}$ is fixed, if $p_{2}<1-p_{3}$, we have
    $\text{P}(Y=1|\tilde{X}=1)>\text{P}(Y=0|\tilde{X}=1)$ and $\text{P}(Y=0|\tilde{X}=0)>\text{P}(Y=1|\tilde{X}=0)$.
    This gives us
    \begin{equation*}
        \begin{split}
          Pe&=\frac{1}{2}\text{P}(\tilde{X}=0|Y=1)+\frac{1}{2}\text{P}(\tilde{X}=1|Y=0)\\
          &=p_{1}(1-p_{2}-p_{3})+\frac{p_{2}+p_{3}}{2}.
        \end{split}
    \end{equation*}
    Otherwise, if $p_{2}>1-p_{3}$, 
    similarly, it follows:
    \begin{equation*}
        Pe=1-p_{1}(1-p_{2}-p_{3})-\frac{p_{2}+p_{3}}{2}.
    \end{equation*}
    Next, we do the convex combination of $p_{2}$ and $p_{3}$, such that
    \begin{equation*}
    \begin{split}
    \text{P}(\tilde{X}=1|X=0)=\text{P}(\tilde{X}=0|X=1)=\frac{p_{2}+p_{3}}{2},\\
    \end{split}
    \end{equation*}
    Since $p_{1}<\frac{1}{2}$ and $p_{1}$ is fixed, by the same process as above, if $p_{2}<1-p_{3}$, we have 
    $Pe=\frac{1}{2}\text{P}(\tilde{X}=0|Y=1)+\frac{1}{2}\text{P}(\tilde{X}=1|Y=0)=p_{1}(1-p_{2}-p_{3})+\frac{p_{2}+p_{3}}{2}$.
   Otherwise, if $p_{2}>1-p_{3}$,
similarly, 
    $Pe=\frac{1}{2}\text{P}(\tilde{X}=0|Y=0)+\frac{1}{2}\text{P}(\tilde{X}=1|Y=1)=1-p_{1}(1-p_{2}-p_{3})-\frac{p_{2}+p_{3}}{2}$.

In conclusion, we notice that $Pe$ remains the same before and after doing convex combination. Since the mutual information is convex function of the transition probability between $X$ and $\tilde{X}$ for fixed $P(X)$ \citep{cover}, mutual information is not increased after doing convex combination while $Pe$ does not increase. This implies that the optimal transition probability should be symmetric.\\
Finally, with this conclusion, we can focus on a symmetric crossover probability $p_2$, namely, 
    $\text{P}(\tilde{X}=1|X=0)=\text{P}(\tilde{X}=0|X=1)=p_{2}$. 
Then, 
    \begin{equation*}
   \begin{split}
    \text{P}(Y=1|\tilde{X}=0)=\text{P}(Y=0|\tilde{X}=1)=p_{1}+p_{2}-2p_{1}p_{2},
    \end{split}   
    \end{equation*}
    Now suppose that $p_{1}<\frac{1}{2}$, and we notice that if we also have $p_{2}<\frac{1}{2}$, then, 
    $\text{P}(Y=0|\tilde{X}=0)>\text{P}(Y=1|\tilde{X}=0)$ and
    $\text{P}(Y=1|\tilde{X}=1)>\text{P}(Y=0|\tilde{X}=1)$.
    This suggests that $Pe=p_{1}+p_{2}-2p_{1}p_{2}$ and the mutual information is given by $1-p_{2}\log\frac{1}{p_{2}}-(1-p_{2})\log\frac{1}{1-p_{2}}$.
\end{proof}
\textbf{Remarks}: Our proof is different from showing that symmetric transition probabilities achieve optimal rate-distortion tradeoff involving  $I (X,\tilde{X})$ and binary distortion between $X$ and $\tilde{X}$. Here we consider the decoding error probability for label $Y$, making our proof arguably more involved.  

\section{Optimization Formulation for Computing Optimal Compression}
\label{sec:general}

Suppose that we have $m$ labels in the label set $\mathcal{Y}$, and we denote them by $y_{1}, y_{2}, \ldots, y_{m}$.  We denote the prior probability for each label as $\text{P}(y_{i})$, $i=1, 2, \ldots, m$. Then these labels generate data sampled from set $\mathcal{X}$ of cardinality $n$. Specifically, the elements in $\mathcal{X}$ are $x_{1}, x_{2}, \ldots, x_{n}$. We denote the transition probability between each label and possible element for data as $\text{P}(x_{j}|y_{i})$, where $i=1, 2, \ldots, m;$ $j=1, 2, \ldots, n$. We want to map (compress) the data to $l$ possible letters in the compressed data set $\tilde{\mathcal{X}}$ of cardinality $l$, which includes $\tilde{x}_{1}, \tilde{x}_{2}, \ldots, \tilde{x}_{l}$ as its elements. Furthermore, we define the transition probability between $x_j$ and compressed data $\tilde{x}_k$ as $\text{P}(\tilde{x}_{k}|x_{j})$, where $j=1, 2, \ldots, n;$ $k=1, 2, \ldots, l.$

Our goal is to minimize the mutual information between $X$ and $\tilde{X}$ by optimizing over the transition probabilities P$(\Tilde x_k|x_{j})$, subject to the constraint that the classification error probability is smaller than a certain threshold $\epsilon$. However, this optimization problem is a non-convex optimization problem. We propose to obtain global optimal solution by dividing this optimization problem into multiple convex optimization problems, based on different MAP decoding rules. 

 We assume that for a given letter $\tilde{x}_{k}$, the MAP rule decodes it to label $\hat{y}_{\tilde{x}_{k}}$, which is from the set ${\mathcal{Y}}$. 
 We notice that there are $m^{l}$ possible MAP maps from $\mathcal{\tilde{X}}$ to $\mathcal{Y}$.  For each MAP decoding rule, we are trying to minimize the mutual information between $X$ and $\tilde{X}$.  
So for a particular MAP decoding rule, minimizing $I(X;\tilde{X})$ is equivalent to the following convex programming:
\begin{equation*}
\begin{split}
\min_{\text{P}(\tilde{x}_{k}|x_{j})}&{I(X;\tilde{X})}\\
\text{subject to} \quad& Pe=\sum_{k=1}^{l}\sum_{y_{i}\neq\hat{y}_{\tilde{x}_{k}}}\sum_{j=1}^{n}{\text{P}(y_{i})\text{P}(x_{j}|y_{i})\text{P}(\tilde{x}_{k}|x_{j})}\leq \epsilon,\quad \\
&\text{P}(\tilde{x}_{k}|x_{j})\geq 0,\quad \forall x_{j}\in\mathcal{X}, \tilde{x}_{k}\in\mathcal{\tilde{X}}  \\
&\sum_{k=1}^{l}{\text{P}(\tilde{x}_{k}|x_{j})=1}, \quad\forall x_{j}\in\mathcal{X}\\
& \sum_{j=1}^{n}{\text{P}(y_{i})\text{P}(x_{j}|y_{i})\text{P}(\tilde{x}_{k}|x_{j})}\\
&\leq \sum_{j=1}^{n}{\text{P}(\hat{y}_{\tilde{x}_{k}})   \text{P}(x_{j}|  \hat{y}_{\tilde{x}_{k}}   )\text{P}(\tilde{x}_{k}|x_{j})}                 \quad \forall k, \forall y_{i}\neq\hat{y}_{\tilde{x}_{k}}
\end{split}
\end{equation*}
where $\epsilon$ is the given error probability tolerance threshold. \emph{We have proved that the minimum objective value among these $m^l$ such convex optimization problems give the globally optimal compression under a constraint on error probability.} This formulation also extends to asymmetrical cost for decoding error.\\

 
\section{NUMERICAL RESULTS}
\label{sec:numerical}
In this section, we present numerical results for characterizing the optimal tradeoff between compression and classification accuracy.\\ 
In Figure 3, we calculate the curve of the allowed mutual information between data ($X$) and compressed data, against the classification error probability for the binary data under the parameters $p_{1}=0.3$. The plot is generated by using the result in Theorem 1. From the plotted curve, we can see that, when the mutual information between data $X$ and compressed data $\tilde{X}$ is allowed to be large, the classification error probability can be reduced, but at the expense of compression efficiency.\\
\begin{figure}
\begin{center}
\includegraphics[width=10cm]{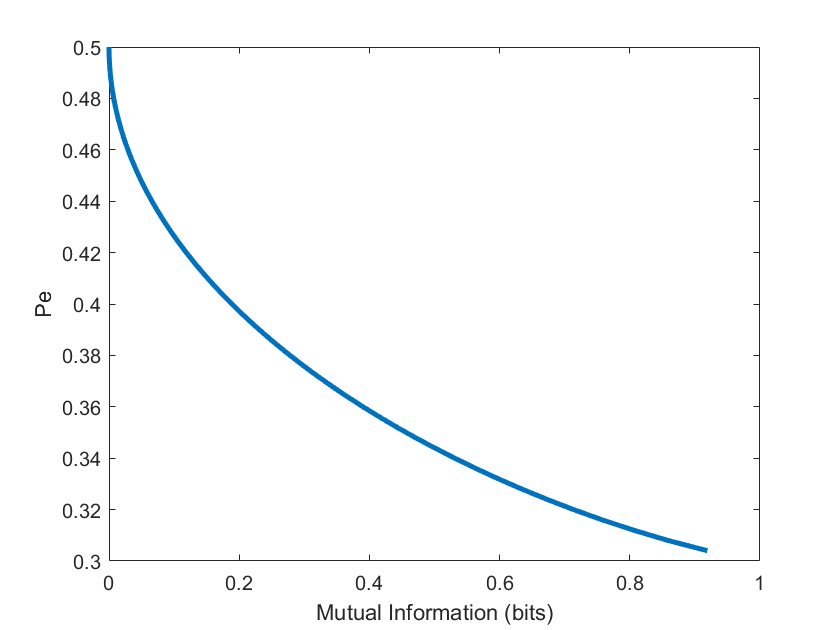}
  \caption{Mutual information between data and compressed data against classification error probability for $p_{1}=0.3.$}
\end{center}
\end{figure}\\
We further consider the case where the costs of decoding to incorrect labels are asymmetrical. In Figure 4, 
we plot the optimal classification cost and data compression trade-off, for a classification task with 3 labels, 4 data letters and 3 compressed data letters, with transition probabilities in the first channel and costs of incorrectly decoding from each label to decoded label shown as follows (the prior probability for each label is $1/3$). Note that when the $c=1$, the cost is equivalent to the decoding error probability. \\
\begin{center}
\begin{tabular}{ | c | c | c | c |} 
  \hline
  $\text{P}(x_{j}|y_{i}) $ & $y_{1}$ & $y_{2}$ & $y_{3}$ \\ 
  \hline
  $x_{1}$ & 0.995 & 0.001 & 0.002 \\ 
  \hline
  $x_{2}$ & 0.001 & 0.996 & 0.002 \\ 
  \hline
  $x_{3}$ & 0.002 & 0.001 & 0.994 \\
  \hline
  $x_{4}$ & 0.002 & 0.002 & 0.002 \\
  \hline
\end{tabular}
\end{center}
\begin{center}
\begin{tabular}{ | c | c | c | c |} 
  \hline
cost & $\hat{y}=y_1$ & $\hat{y}=y_2$ & $\hat{y}=y_3$ \\ 
  \hline
  $y_{1}$ & 0 & c & c \\ 
  \hline
  $y_{2}$ & 1 & 0 & 1 \\ 
  \hline
  $y_{3}$ & 1 & 1 & 0 \\
  \hline
\end{tabular}
\end{center}
\begin{figure}
\begin{center}
\includegraphics[width=10cm]{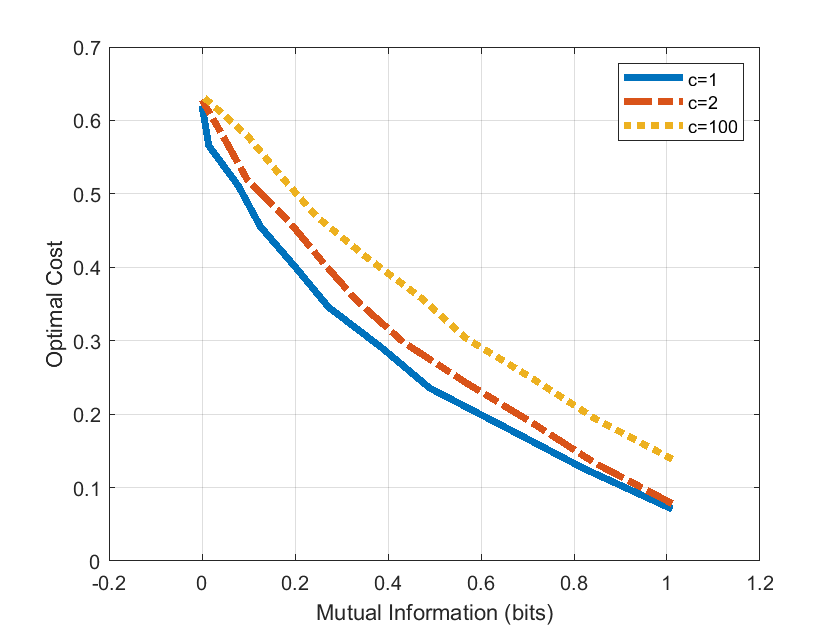}
  \caption{Optimal cost against mutual information between data and compressed data for different $c$.}
\end{center}
\end{figure}
Next, we consider the case with 3 labels, 3 data letters and 2 compressed data letters where costs of incorrectly decoding from each label to decoded label and transition probabilities between label and data are shown in the following tables.(the prior probability for each label is $1/4$, $1/4$ and $1/2$)
\begin{center}
\begin{tabular}{ | c | c | c | c |} 
  \hline
  $\text{P}(x_{j}|y_{i}) $ & $y_{1}$ & $y_{2}$ & $y_{3}$ \\ 
  \hline
  $x_{1}$ & 0.9 & 0.1 & 0.05 \\ 
  \hline
  $x_{2}$ & 0.1 & 0.9 & 0.05 \\ 
  \hline
  $x_{3}$ & 0 & 0 & 0.9 \\
  \hline
\end{tabular}
\end{center}
\begin{center}
\begin{tabular}{ | c | c | c | c |} 
  \hline
   cost & $\hat{y}=y_1$ & $\hat{y}=y_2$ & $\hat{y}=y_3$ \\ 
  \hline
  $y_{1}$ & 0 & 1 & 1 \\ 
  \hline
  $y_{2}$ & 1 & 0 & 1 \\ 
  \hline
  $y_{3}$ & 0.0001 & 0.0001 & 0 \\
  \hline
\end{tabular}
\end{center}
\begin{figure}
\begin{center}
\includegraphics[width=10cm]{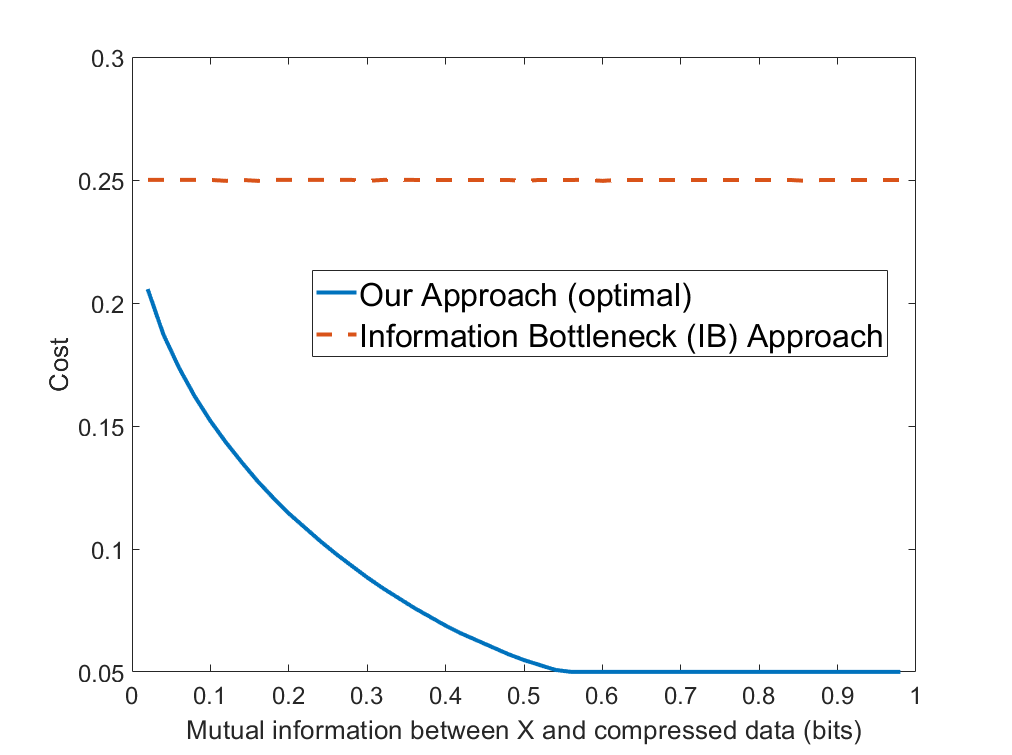}
  \caption{Optimal cost and compression tradeoff for our approach, and comparison with the performance of information bottleneck approach.}
\end{center}
\end{figure}
In Figure 5, compared with Information Bottleneck Principle (IBP, which directly maximizes mutual information between label and compressed data), we get a curve of decoding cost against the mutual information between data ($X$) and compressed data ($\tilde{X}$). As we can see, our newly proposed approach can significantly outperform the IBP approach in achieving minimum decoding cost and highest compression efficiency. The reason is that the information bottleneck approach was not optimized for minimizing the cost.

\vfill\pagebreak


\bibliography{refs}

\begin{thebibliography}{12}
\providecommand{\natexlab}[1]{#1}
\providecommand{\url}[1]{\texttt{#1}}
\expandafter\ifx\csname urlstyle\endcsname\relax
  \providecommand{\doi}[1]{doi: #1}\else
  \providecommand{\doi}{doi: \begingroup \urlstyle{rm}\Url}\fi

\bibitem[Bardera et~al.(2009)Bardera, Rigau, Boada, Feixas, and Sbert]{bottle}
Anton Bardera, Jaume Rigau, Imma Boada, Miquel Feixas, and Mateu Sbert.
\newblock Image segmentation using information bottleneck method.
\newblock \emph{IEEE Transactions on Image Processing}, 18\penalty0
  (7):\penalty0 1601--1612, 2009.
\newblock \doi{10.1109/TIP.2009.2017823}.

\bibitem[Bishop(2006)]{bishop:2006:PRML}
Christopher~M. Bishop.
\newblock \emph{Pattern Recognition and Machine Learning}.
\newblock Springer, 2006.

\bibitem[Calderbank et~al.(2009)Calderbank, Jafarpour, and
  Schapire]{Calderbank09compressedlearning:}
Robert Calderbank, Sina Jafarpour, and Robert Schapire.
\newblock Compressed learning: Universal sparse dimensionality reduction and
  learning in the measurement domain.
\newblock Technical report, 2009.

\bibitem[Cheng et~al.()Cheng, Pavone, Katti, Chinchali, and
  Tang]{cheng2021datasharing}
Jiangnan Cheng, Marco Pavone, Sachin Katti, Sandeep Chinchali, and Ao~Tang.
\newblock Data sharing and compression for cooperative networked control.
\newblock \emph{accepted to NeurIPS 2021}.

\bibitem[Cover and Thomas(2006)]{cover}
Thomas~M. Cover and Joy~A. Thomas.
\newblock \emph{Elements of Information Theory (Wiley Series in
  Telecommunications and Signal Processing)}.
\newblock Wiley-Interscience, USA, 2006.
\newblock ISBN 0471241954.

\bibitem[Geiger and Kubin(2020)]{entropy}
Bernhard Geiger and Gernot Kubin.
\newblock Information bottleneck: Theory and applications in deep learning.
\newblock \emph{Entropy}, 22:\penalty0 1408, 12 2020.
\newblock \doi{10.3390/e22121408}.

\bibitem[Murphy(2013)]{murphy2013machine}
Kevin~P. Murphy.
\newblock \emph{Machine learning : a probabilistic perspective}.
\newblock MIT Press, Cambridge, Mass. [u.a.], 2013.
\newblock ISBN 9780262018029 0262018020.
\newblock URL
  \url{https://www.amazon.com/Machine-Learning-Probabilistic-Perspective-Computation/dp/0262018020/ref=sr_1_2?ie=UTF8&qid=1336857747&sr=8-2}.

\bibitem[Raginsky(2009)]{raginskyDBLP:journals/corr/abs-0901-1905}
Maxim Raginsky.
\newblock Achievability results for statistical learning under communication
  constraints.
\newblock \emph{CoRR}, abs/0901.1905, 2009.
\newblock URL \url{http://arxiv.org/abs/0901.1905}.

\bibitem[Raginsky(2013)]{raginsky6353589}
Maxim Raginsky.
\newblock Empirical processes, typical sequences, and coordinated actions in
  standard borel spaces.
\newblock \emph{IEEE Transactions on Information Theory}, 59\penalty0
  (3):\penalty0 1288--1301, 2013.
\newblock \doi{10.1109/TIT.2012.2227669}.

\bibitem[Tishby and Zaslavsky(2015)]{tishby2015learning}
Naftali Tishby and Noga Zaslavsky.
\newblock Deep learning and the information bottleneck principle.
\newblock \emph{CoRR}, abs/1503.02406, 2015.
\newblock URL
  \url{http://dblp.uni-trier.de/db/journals/corr/corr1503.html#TishbyZ15}.

\bibitem[Tishby et~al.(2001)Tishby, Pereira, and Bialek]{Tisbybottle}
Naftali Tishby, Fernando Pereira, and William Bialek.
\newblock The information bottleneck method.
\newblock \emph{Proceedings of the 37th Allerton Conference on Communication,
  Control and Computation}, 49, 07 2001.

\bibitem[Zisselman et~al.(2018)Zisselman, Adler, and Elad]{ZISSELMAN20183}
E.~Zisselman, A.~Adler, and M.~Elad.
\newblock Chapter 1 - compressed learning for image classification: A deep
  neural network approach.
\newblock In Ron Kimmel and Xue-Cheng Tai, editors, \emph{Processing, Analyzing
  and Learning of Images, Shapes, and Forms: Part 1}, volume~19 of
  \emph{Handbook of Numerical Analysis}, pages 3--17. Elsevier, 2018.
\newblock \doi{https://doi.org/10.1016/bs.hna.2018.08.002}.
\newblock URL
  \url{https://www.sciencedirect.com/science/article/pii/S1570865918300024}.

\end{thebibliography}
\end{document}